\documentclass[12pt,reqno]{amsart}

\usepackage{graphicx}
\graphicspath{ {./images/} }
\usepackage{bm}

\usepackage{color}

\usepackage{amssymb}

\usepackage{amsfonts}

\usepackage{amsmath}

\usepackage[all]{xy}

\newtheorem{theorem}{Theorem}[section]

\newtheorem{corollary}[theorem]{Corollary}

\newtheorem{lemma}[theorem]{Lemma}

\newtheorem{Definition}[theorem]{Definition}

\newtheorem{Example}[theorem]{Example}

\newtheorem{Remark}[theorem]{Remark}

\newenvironment{remark}{\begin{Remark}\begin{em}}{\end{em}\end{Remark}}

\DeclareMathOperator{\tr}{tr}

\address{Miran Jeong, Jinmi Hwang, and Sejong Kim \\ Department of Mathematics, Chungbuk National University, Cheongju 28644, Korea}
\email{jmr4006@chungbuk.ac.kr, jinmi0401@chungbuk.ac.kr, skim@chungbuk.ac.kr}

\begin{document}

\title[Right mean for the $\alpha-z$ Bures-Wasserstein quantum divergence]{Right mean for the $\alpha-z$ Bures-Wasserstein quantum divergence}

\author{Miran Jeong, Jinmi Hwang, and Sejong Kim}

\date{}
\maketitle

\begin{abstract}
A new quantum divergence induced from the $\alpha-z$ R\'{e}nyi relative entropy, called the $\alpha-z$ Bures-Wasserstein quantum divergence, has been recently introduced. We investigate in this paper properties of the right mean, which is a unique minimizer of the weighted sum of $\alpha-z$ Bures-Wasserstein quantum divergences to each points. Many interesting operator inequalities of the right mean with the matrix power mean including the Cartan mean are presented. Moreover, we verify the trace inequality with the Wasserstein mean and provide bounds for the Hadamard product of two right means.
\vspace{5mm}


\end{abstract}

\section{Introduction}

The Fr\'{e}chet mean on a metric space $(X, d)$ is the least squares mean, which is a minimizer of the weighted sum of squared distances to each variable:
\begin{displaymath}
\underset{x \in X}{\arg \min} \sum_{j=1}^{n} w_{j} d^{2}(a_{j}, x),
\end{displaymath}
where $a_{1}, \dots, a_{n} \in X$ and $\omega = (w_{1}, \dots, w_{n})$ is a positive probability vector. Although it may be difficult to determine the existence of the Fr\'{e}chet mean on given a metric space in general, it has been shown on a non-positive curvature metric space or Hadamard space (a complete metric space satisfying semi-parallelogram law) that the Fr\'{e}chet mean uniquely exists. The canonical example of the Hadamard space is the open convex cone $\mathbb{P}_{m}$ of all $m \times m$ positive definite Hermitian matrices equipped with the Riemannian trace metric $d_{R}(A, B) = \Vert \log A^{-1/2} B A^{-1/2} \Vert_{2}$. Moreover, the Fr\'{e}chet mean on this Hadamard space $(\mathbb{P}_{m}, d_{R})$ is the weighted Cartan mean (Karcher mean)
\begin{displaymath}
\Lambda(\omega; A_{1}, \dots, A_{n}) = \underset{X \in \mathbb{P}_{m}}{\arg \min} \sum_{j=1}^{n} w_{j} d_{R}^{2}(A_{j}, X).
\end{displaymath}
Many interesting approaches to the weighted Cartan mean such as the family of power means \cite{LP} and the deterministic sequence \cite{Ho, LP14} have been developed.

The Bures-Wasserstein metric on the open convex cone $\mathbb{P}_{m}$ is given by
\begin{displaymath}
d_{W}(A, B) = \left[ \tr \left( \frac{A + B}{2} \right) - \tr (A^{1/2} B A^{1/2})^{1/2} \right]^{1/2},
\end{displaymath}
which coincides with the $L_{2}$-Wasserstein distance of two Gaussian measures with mean zero and covariance matrices $A$ and $B$ \cite{ABCM}. On the other hand, it does not give us the non-positive curvature metric on $\mathbb{P}_{m}$. Nevertheless, the objective function $\displaystyle X \mapsto \sum_{j=1}^{n} w_{j} d_{W}^{2}(A_{j}, X)$ is strictly convex \cite{BJL}, so the least squares mean uniquely exists. We call it the weighted Wasserstein mean
\begin{displaymath}
\Omega(\omega; A_{1}, \dots, A_{n}) = \underset{X \in \mathbb{P}_{m}}{\arg \min} \sum_{j=1}^{n} w_{j} d_{W}^{2}(A_{j}, X).
\end{displaymath}
Many remarkable properties such as the iteration approach using optimal transport maps \cite{ABCM}, the extended Lie-Trotter-Kato formula and operator inequalities \cite{HK, KL20} have been established.

Divergence is introduced as a distance-like function that does not necessarily satisfy the symmetry nor the triangle inequality. In many literatures, it is a generalization of squared distance. It is originated in statistics, probability theory and information theory, and recently plays important roles in many practical areas such as signal processing \cite{VO}, medical image analysis \cite{PMV}, econometrics \cite{Ul}, and clustering algorithms \cite{AB, BMDG, DT}. Since the divergence $D$ on a set $X$ is not symmetric in general, we have two kinds of weighted means, called respectively the right mean and left mean,
\begin{center}
$\displaystyle \underset{x \in X}{\arg \min} \sum_{j=1}^{n} w_{j} D(a_{j}, x)$ \ and \ $\displaystyle \underset{x \in X}{\arg \min} \sum_{j=1}^{n} w_{j} D(x, a_{j})$
\end{center}
for given $a_{1}, \dots, a_{n} \in X$. Chebbi and Moakher \cite{CM} have first introduced the right mean and left mean of positive definite Hermitian matrices for the log-determinant $\alpha$-divergence:
\begin{displaymath}
D_{LD}^{\alpha} (A, B) = \frac{4}{1 - \alpha^{2}} \left[ \log \det \left( \frac{1 - \alpha}{2} A + \frac{1 + \alpha}{2} B \right) - \log (\det A)^{\frac{1 - \alpha}{2}} (\det B)^{\frac{1 + \alpha}{2}} \right]
\end{displaymath}
for $\alpha \in (-1, 1)$ and $A, B \in \mathbb{P}_{m}$.

A quantum divergence is a smooth function $\Phi: \mathbb{P}_{m} \times \mathbb{P}_{m} \to \mathbb{R}$ satisfying
\begin{itemize}
\item[(i)] $\Phi(A, B) \geq 0$, and equality holds if and only if $A = B$,
\item[(ii)] the first derivative $D \Phi$ with respect to the second variable vanishes on the diagonal, that is,
\begin{displaymath}
D \Phi(A, X) |_{X=A} = 0,
\end{displaymath}
\item[(iii)] the second derivative is positive on the diagonal, that is,
\begin{displaymath}
D^{2} \Phi(A, X) |_{X=A} (Y,Y) \geq 0
\end{displaymath}
for any Hermitian matrix $Y$.
\end{itemize}
See \cite[Sections 1.2 and 1.3]{Am} for more information. The well-known examples are the Kullback-Leibler divergence, and the Bregman divergence corresponding to a strictly convex differentiable function. Bhatia, Gaubert, and Jain \cite{BGJ} have recently introduced different kinds of quantum divergences of the form
\begin{displaymath}
\Phi(A, B) = \tr \mathcal{A}(A, B) - \tr \mathfrak{G}(A, B),
\end{displaymath}
where $\displaystyle \mathcal{A}(A, B) = \frac{A + B}{2}$ is the two-variable arithmetic mean and $\mathfrak{G}(A, B)$ is a matrix version of two-variable geometric mean such as the Riemannian geodesic midpoint $\mathfrak{G}(A, B) = A^{1/2} (A^{-1/2} B A^{-1/2})^{1/2} A^{1/2} = A \# B$, and the log-Euclidean mean $\displaystyle \mathfrak{G}(A, B) = \exp \left( \frac{\log A + \log B}{2} \right)$. Moreover, they have also provided the left means corresponding to such quantum divergences and their characterizations.

A new quantum divergence, called the $\alpha-z$ Bures-Wasserstein quantum divergence, has been recently introduced \cite{DLVV}: for $0 < \alpha \leq z < 1$
\begin{equation} \label{E:divergence}
\Phi_{\alpha, z}(A, B) = \tr ((1 - \alpha) A + \alpha B) - \tr  \left( A^{\frac{1 - \alpha}{2z}} B^{\frac{\alpha}{z}} A^{\frac{1 - \alpha}{2z}} \right)^{z}.
\end{equation}
For $\alpha = z = \frac{1}{2}$, $\Phi_{\alpha, z}(A, B)$ coincides with the Bures-Wasserstein metric of $A$ and $B$ \cite{BJL}. It has been shown that the quantum divergence $\Phi_{\alpha, z}$ is invariant under any unitary congruence transformation and tensor product with another density matrix. Also, the right mean for the $\alpha-z$ Bures-Wasserstein quantum divergence $\Phi_{\alpha, z}$ exists uniquely, so we call it the $\alpha-z$ weighted right mean. Moreover, it coincides with the unique positive definite solution of the equation
\begin{displaymath}
X = \sum_{j=1}^{n} w_{j} \left( X^{\frac{\alpha}{2z}} A_{j}^{\frac{1 - \alpha}{z}} X^{\frac{\alpha}{2z}} \right)^{z}.
\end{displaymath}
In this paper we verify interesting operator inequalities of the $\alpha-z$ weighted right mean with matrix power means including the weighted arithmetic and harmonic means. It provides the log-majorization properties among the $\alpha-z$ weighted right mean, arithmetic mean, Cartan mean, and harmonic mean. Furthermore, we show the trace inequality with the Wasserstein mean and bounds for Hadamard product of two $\alpha-z$ weighted right means.

\section{Right mean for the $\alpha-z$ Bures-Wasserstein quantum divergence}

Let $M_{m, k}$ be the set of all $m \times k$ complex matrices, and we simply denote as $M_{m} := M_{m,m}$. Let $\mathbb{H}_{m}$ be the real vector space of all $m \times m$ Hermitian matrices. Let $\mathbb{P}_{m}$ be the open convex cone of all $m \times m$ positive definite Hermitian matrices. For $A, B \in \mathbb{P}_{m}$, $0 \leq \alpha \leq 1$ and $z > 0$
\begin{displaymath}
Q_{\alpha, z}(A, B) = \left( A^{\frac{1 - \alpha}{2z}} B^{\frac{\alpha}{z}} A^{\frac{1 - \alpha}{2z}} \right)^{z}
\end{displaymath}
is the matrix version of the $\alpha-z$ R\'{e}nyi relative entropy \cite{AD}. Especially, $Q_{\alpha, \alpha}(A, B)$ is known as the sandwiched quasi-relative entropy \cite{WWY}. Recently, a new quantum divergence, called the $\alpha-z$ Bures-Wasserstein quantum divergence, has been introduced \cite{DLVV}: for $0 < \alpha \leq z < 1$
\begin{displaymath}
\Phi_{\alpha, z}(A, B) = \tr ((1 - \alpha) A + \alpha B) - \tr Q_{\alpha, z}(A, B).
\end{displaymath}
For $\alpha = z = \frac{1}{2}$, $\Phi_{\alpha, z}(A, B)$ coincides with the Bures-Wasserstein metric of $A$ and $B$ \cite{BJL}.

For an $n$-tuple $\mathbb{A} = (A_{1}, \dots, A_{n}) \in \mathbb{P}_{m}^{n}$ and a positive probability vector $\omega = (w_{1}, \dots, w_{n})$, we consider the minimization problem
\begin{equation} \label{E:minimization}
\underset{X \in \mathbb{P}_{m}}{\arg \min} \sum_{j=1}^{n} w_{j} \Phi_{\alpha, z}(A_{j}, X).
\end{equation}
By the strict concavity of the map $\mathbb{P}_{m} \ni A \mapsto \tr A^{t}$ for $t \in [0,1]$ from \cite{BJL18} and by the linearity and monotonicity of the congruence transformation, the map $\mathbb{P}_{m} \ni X \mapsto \tr Q_{\alpha, z}(A, X)$ is strictly concave. Then the objective function $\displaystyle F(X) = \sum_{j=1}^{n} w_{j} \Phi_{\alpha, z}(A_{j}, X)$ is strictly convex, so the minimization \eqref{E:minimization} has a unique solution in $\mathbb{P}_{m}$. By vanishing the gradient of $F(X)$, we obtain from \cite{DLVV} that the unique solution coincides with the unique positive definite solution of the matrix equation
\begin{equation} \label{E:Renyi}
X = \sum_{j=1}^{n} w_{j} Q_{1 - \alpha, z}(X, A_{j}).
\end{equation}
We write such a unique minimizer of \eqref{E:minimization} as $\mathcal{R}_{\alpha, z}(\omega; \mathbb{A})$ and call the $\alpha-z$ weighted right mean.

\begin{lemma} \label{L:equation}
Let $\mathbb{A} = (A_{1}, \dots, A_{n}) \in \mathbb{P}_{m}^{n}$ and let $\omega = (w_{1}, \dots, w_{n})$ be a positive probability vector. Then for $0 < \alpha \leq z < 1$ the $\alpha-z$ weighted right mean $\mathcal{R}_{\alpha, z}(\omega; \mathbb{A})$ is the unique positive definite solution of the matrix equation
\begin{equation} \label{E:matrix equation}
X^{1 - \frac{\alpha}{z}} = \sum_{j=1}^{n} w_{j} X^{- \frac{\alpha}{z}} \#_{z} A_{j}^{\frac{1 - \alpha}{z}}.
\end{equation}
\end{lemma}

\begin{proof}
Let $X = \mathcal{R}_{\alpha, z}(\omega; \mathbb{A})$. By \eqref{E:Renyi}
\begin{displaymath}
X = \sum_{j=1}^{n} w_{j} \left( X^{\frac{\alpha}{2z}} A_{j}^{\frac{1 - \alpha}{z}} X^{\frac{\alpha}{2z}} \right)^{z}.
\end{displaymath}
Taking congruence transformation by $X^{- \frac{\alpha}{2z}}$ yields
\begin{displaymath}
X^{1 - \frac{\alpha}{z}} = \sum_{j=1}^{n} w_{j} X^{- \frac{\alpha}{2z}} \left( X^{\frac{\alpha}{2z}} A_{j}^{\frac{1 - \alpha}{z}} X^{\frac{\alpha}{2z}} \right)^{z} X^{- \frac{\alpha}{2z}} = \sum_{j=1}^{n} w_{j} X^{- \frac{\alpha}{z}} \#_{z} A_{j}^{\frac{1 - \alpha}{z}}.
\end{displaymath}
\end{proof}

Let $\Delta_{n}$ be the simplex of positive probability vectors in $\mathbb{R}^{n}$ convexly spanned by the unit coordinate vectors. Let $\mathbb{A} = (A_{1},\dots, A_{n}) \in \mathbb{P}_{m}^{n}$, $\omega = (w_{1},\dots, w_{n}) \in \Delta_{n}$, $\sigma \in S^{n}$ a permutation on $n$-letters, $p \in \mathbb{R}$ and $M \in \textrm{GL}_{m}$, the general linear group. For convenience, we denote as
\begin{displaymath}
\begin{split}
\omega_{\sigma} & := (w_{\sigma(1)},\dots, w_{\sigma(n)}), \\
\mathbb{A}_{\sigma} & := (A_{\sigma(1)},\dots, A_{\sigma(n)}), \\
\mathbb{A}^{p} & := (A_{1}^{p},\dots, A_{n}^{p}), \\
M \mathbb{A} M^{*} & := (M A_{1} M^{*},\dots, M A_{n} M^{*}),
\end{split}
\end{displaymath}
and
\begin{displaymath}
\begin{split}
\hat{\omega} & := \frac{1}{1 - w_{n}}(w_{1}, \dots, w_{n-1}) \in \Delta_{n-1}, \\
\omega^{(k)} & := \frac{1}{k} (\underbrace{w_{1}, \dots, w_{n}}, \dots, \underbrace{w_{1}, \dots, w_{n}}) \in \Delta_{nk}, \\
\mathbb{A}^{(k)} & := (\underbrace{A_{1}, \dots, A_{n}}, \dots, \underbrace{A_{1}, \dots, A_{n}}) \in \mathbb{P}_{m}^{nk},
\end{split}
\end{displaymath}
of which number of tuples is $k \in \mathbb{N}$.

\begin{lemma} \cite{HJK} \label{L:properties}
The $\alpha-z$ weighted right mean $\mathcal{R}_{\alpha, z}$ satisfies the following:
\begin{itemize}
\item[(1)] $\displaystyle \mathcal{R}_{\alpha, z}(\omega; \mathbb{A}) = \left( \sum_{j=1}^{n} w_{j} A_{j}^{1 - \alpha} \right)^{\frac{1}{1 - \alpha}}$ if $A_{j}$'s commute;
\item[(2)] $\mathcal{R}_{\alpha, z}(\omega; c \mathbb{A}) = c \mathcal{R}_{\alpha, z}(\omega; \mathbb{A})$ for any $c > 0$;
\item[(3)] $\mathcal{R}_{\alpha, z}(\omega_{\sigma}; \mathbb{A}_{\sigma}) = \mathcal{R}_{\alpha, z}(\omega; \mathbb{A})$ for any permutation $\sigma$ on $\{ 1, \dots, n \}$;
\item[(4)] $\mathcal{R}_{\alpha, z}(\omega^{(k)}; \mathbb{A}^{(k)}) = \mathcal{R}_{\alpha, z}(\omega; \mathbb{A})$ for any natural number $k$;
\item[(5)] $\mathcal{R}_{\alpha, z}(\omega; U \mathbb{A} U^{*}) = U \mathcal{R}_{\alpha, z}(\omega; \mathbb{A}) U^{*}$ for any unitary matrix $U$;
\item[(6)] $\displaystyle \det \mathcal{R}_{\alpha, z}(\omega; \mathbb{A}) \geq \prod_{j=1}^{n} (\det A_{j})^{w_{j}}$, and equality holds if and only if $A_{1} = \cdots = A_{n}$;
\item[(7)] $\displaystyle X = \mathcal{R}_{\alpha, z}(\omega; A_1, \dots, A_{n-1}, X)$ implies that $X = \mathcal{R}_{\alpha, z}(\hat\omega; A_1, \dots, A_{n-1})$;
\item[(8)] $\mathcal{R}_{\alpha, z}(\omega;\mathbb{A}) = \mathcal{R}_{\alpha, z}\left( \sum \limits_{j=1}^{k} w_{j}, w_{k+1}, \dots, w_{n}; A_1, A_{k+1}, \dots, A_n \right)$ if $A_1 = \cdots = A_k$ for $1 \leq k < n$.
\end{itemize}
\end{lemma}

\begin{lemma} \label{L:boundedness}
Let $\omega = (w_{1}, \dots, w_{n}) \in \Delta_{n}$, and let $\mathbb{A} = (A_{1}, \dots, A_{n}) \in \mathbb{P}_{m}^{n}$ satisfying $a I \leq A_{j} \leq b I$ for all $j$ and some $0 < a \leq b$.  Then $a I \leq \mathcal{R}_{\alpha, z}(\omega; \mathbb{A}) \leq b I$ for $0 < \alpha \leq z < 1$.
\end{lemma}

\begin{proof}
Let $X = \mathcal{R}_{\alpha, z}(\omega; \mathbb{A})$ for $0 < \alpha \leq z < 1$. Assume that $a I \leq A_{j} \leq b I$ for all $j$ and some $0 < a \leq b$. Then $a^{\frac{1-\alpha}{z}} I \leq A_{j}^{\frac{1-\alpha}{z}} \leq b^{\frac{1-\alpha}{z}} I$. By the monotonicity of two-variable weighted geometric mean
\begin{displaymath}
a^{1-\alpha} X^{-\frac{(1-z) \alpha}{z}} \leq X^{- \frac{\alpha}{z}} \#_{z} A_{j}^{\frac{1 - \alpha}{z}} \leq b^{1-\alpha} X^{-\frac{(1-z) \alpha}{z}}.
\end{displaymath}
Summing up for all $j$ and applying Lemma \ref{L:equation} yield $a^{1-\alpha} X^{-\frac{(1-z) \alpha}{z}} \leq X^{1 - \frac{\alpha}{z}} \leq b^{1-\alpha} X^{-\frac{(1-z) \alpha}{z}}$. Taking the congruence transformation by $X^{\frac{(1-z) \alpha}{2z}}$ and simplifying the powers we obtain the conclusion.
\end{proof}

\section{Operator inequalities with matrix power means}

Let $\mathbb{A} = (A_{1}, \dots, A_{n}) \in \mathbb{P}_{m}^{n}$ and let $\omega = (w_{1}, \dots, w_{n}) \in \Delta_{n}$. The matrix power mean $P_{t}(\omega; \mathbb{A})$ for $t \in (0,1]$ is defined in \cite{LP} as the unique solution $X \in \mathbb{P}_{m}$ of the following equation
\begin{displaymath}
X = \sum_{j=1}^{n} w_{j} X \#_{t} A_{j}.
\end{displaymath}
Indeed, the map $g: \mathbb{P}_{m} \to \mathbb{P}_{m}, \displaystyle  g(X) = \sum_{j=1}^{n} w_{j} X \#_{t} A_{j}$ for $t \in (0,1]$ is an operator monotone function and a strict contraction for the Thompson metric $d_{T}(A, B) = \Vert \log A^{-1/2} B A^{-1/2} \Vert$, where $\Vert \cdot \Vert$ denotes the operator norm. Therefore, by the Banach fixed point theorem
\begin{displaymath}
\lim_{k \to \infty} g^{k}(Z) = P_{t}(\omega; \mathbb{A}) \hspace{1cm} \textrm{ for any } Z \in \mathbb{P}_{m}.
\end{displaymath}
For $t \in [-1,0)$ we define $P_{t}(\omega; \mathbb{A}) = P_{-t}(\omega; \mathbb{A}^{-1})^{-1}$. Note that
\begin{displaymath}
\begin{split}
P_{1}(\omega; \mathbb{A}) & = \sum_{j=1}^{n} w_{j} A_{j} = \mathcal{A}(\omega; \mathbb{A}), \\
P_{-1}(\omega; \mathbb{A}) & = \left[ \sum_{j=1}^{n} w_{j} A_{j}^{-1} \right]^{-1} = \mathcal{H}(\omega; \mathbb{A})
\end{split}
\end{displaymath}
are the weighted arithmetic and harmonic means, respectively. The most remarkable consequence of matrix power means is that matrix power means $P_{t}(\omega; A_{1}, \dots, A_{n})$ converges to the Cartan mean $\Lambda(\omega; A_{1}, \dots, A_{n})$ as $t \to 0$. This plays an important role to construct the Karcher mean of positive invertible operators: see \cite{LL14}. Furthermore, the power mean interpolates monotonically the weighted arithmetic, Cartan, and harmonic means in the sense that for $0 \leq s \leq t \leq 1$
\begin{equation} \label{E:para-monotonicity}
\mathcal{H} = P_{-1} \leq P_{-t} \leq P_{-s} \leq \cdots \leq \Lambda \leq \cdots \leq P_{s} \leq P_{t} \leq P_{1} = \mathcal{A}.
\end{equation}

\begin{theorem} \label{T:A-R ineq}
For $\frac{1}{2} \leq z \leq 1$,
\begin{displaymath}
\mathcal{R}_{\alpha, z}(\omega; \mathbb{A})^{\frac{1 - \alpha}{z}} \leq \mathcal{A}(\omega; \mathbb{A}^{\frac{1 - \alpha}{z}}).
\end{displaymath}
\end{theorem}

\begin{proof}
Let $X = \mathcal{R}_{\alpha, z}(\omega; \mathbb{A})$ and $\frac{1}{2} \leq z \leq 1$. Then $\displaystyle X = \sum_{j=1}^{n} w_{j} \left(X^{\frac{\alpha}{2 z}} A_{j}^{\frac{1 - \alpha}{z}} X^{\frac{\alpha}{2 z}}\right)^{z}$. Taking $1/z$-power on both sides and applying the convexity of the map $A \mapsto A^{r}$ for $1 \leq r \leq 2$,
\begin{displaymath}
X^{\frac{1}{z}} = \left[ \sum_{j=1}^{n} w_{j} \left(X^{\frac{\alpha}{2 z}} A_{j}^{\frac{1 - \alpha}{z}} X^{\frac{\alpha}{2 z}}\right)^{z} \right]^{\frac{1}{z}} \leq \sum_{j=1}^{n} w_{j} X^{\frac{\alpha}{2 z}} A_{j}^{\frac{1 - \alpha}{z}} X^{\frac{\alpha}{2 z}} = X^{\frac{\alpha}{2 z}} \left( \sum_{j=1}^{n} w_{j} A_{j}^{\frac{1 - \alpha}{z}} \right) X^{\frac{\alpha}{2 z}}.
\end{displaymath}
Thus, we have $\displaystyle X^{\frac{1 - \alpha}{z}} \leq \sum_{j=1}^{n} w_{j} A_{j}^{\frac{1 - \alpha}{z}}$.
\end{proof}


\begin{theorem} \label{T:inequalities-2}
Let $0 < \alpha \leq z < 1.$ If $\mathcal{R}_{\alpha, z}(\omega; \mathbb{A}) \leq I$ then
\begin{displaymath}
\mathcal{R}_{\alpha, z}(\omega; \mathbb{A})^{1 - \frac{\alpha}{z}} \geq \mathcal{A}(\omega; \mathbb{A}^{1 - \alpha}).
\end{displaymath}
If $\mathcal{R}_{\alpha, z}(\omega; \mathbb{A}) \geq I$ then the reverse inequality holds.
\end{theorem}

\begin{proof}
Assume $X = \mathcal{R}_{\alpha, z}(\omega; \mathbb{A}) \leq I$ for $0 < \alpha \leq z < 1$. Then $X^{-\frac{\alpha}{z}} \geq I$, so by the matrix equation \eqref{E:matrix equation}
\begin{displaymath}
X^{1-\frac{\alpha}{z}} = \sum_{j=1}^{n} w_{j} X^{-\frac{\alpha}{z}} \#_{z} A_{j}^{\frac{1 - \alpha}{z}} \geq \sum_{j=1}^{n} w_{j} I\#_{z} A_{j}^{\frac{1 - \alpha}{z}} = \sum_{j=1}^{n} w_{j} A_{j}^{1-\alpha}.
\end{displaymath}
The second inequality follows from the monotonicity of the weighted geometric mean.

For the case that $\mathcal{R}_{\alpha, z}(\omega; \mathbb{A}) \geq I$ we can prove the reverse inequality by the similar method as above.
\end{proof}

\begin{theorem} \label{T:Renyi-power}
Let $0 < \alpha \leq z < 1$. If $\mathcal{R}_{\alpha, z}(\omega; \mathbb{A}) \geq I$ then
\begin{displaymath}
\mathcal{R}_{\alpha, z}(\omega; \mathbb{A})^{1-\frac{\alpha}{z}} \leq P_{z}(\omega; \mathbb{A}^{\frac{1 - \alpha}{z}}).
\end{displaymath}
If $\mathcal{R}_{\alpha, z}(\omega; \mathbb{A}) \leq I$ then the reverse inequality holds.
\end{theorem}

\begin{proof}
Assume that $X = \mathcal{R}_{\alpha, z}(\omega; \mathbb{A})^{-\frac{\alpha}{z}}$ for $0 < \alpha \leq z < 1$. Then $X^{-\frac{z}{\alpha}}$ satisfies the following equation from Lemma \ref{L:equation}
\begin{displaymath}
X^{1 - \frac{z}{\alpha}} = \sum_{j=1}^{n} w_{j} X \#_{z} A_{j}^{\frac{1 - \alpha}{z}}.
\end{displaymath}
By assumption $X \leq I$, so $X \leq X^{1 - \frac{z}{\alpha}}$.
Thus we get
\begin{displaymath}
X \leq X^{1 - \frac{z}{\alpha}} = \sum_{j=1}^{n} w_{j} X \#_{z} A_{j}^{\frac{1 - \alpha}{z}} =: f(X).
\end{displaymath}
Since the map $f$ is operator monotone, $X \leq f(X) \leq f^{2}(X) \leq \cdots \leq f^{k}(X)$ for all $k \in \mathbb{N}$. Since $f^{k}(X)$ converges to the power mean $P_{z}(\omega; \mathbb{A}^{\frac{1 - \alpha}{z}})$ as $k \to \infty$, we obtain the desired inequality.

By the similar argument for $\mathcal{R}_{\alpha, z}(\omega; \mathbb{A}) \leq I$, we can prove that the reverse inequality is satisfied.
\end{proof}

\begin{remark}
Note that $A^{1 - \alpha}$ and $A^{\frac{1 - \alpha}{z}}$ for any $A \in \mathbb{P}_{m}$ and $0 < \alpha \leq z < 1$ can not be compared in general, but only when $A \geq I$ or $A \leq I$. It means that Theorem \ref{T:inequalities-2} and Theorem \ref{T:Renyi-power} are different results.
\end{remark}

The theory of majorization and log-majorization plays an important role in matrix inequalities of eigenvalues, singular values and matrix norm. Let $\mathbf{x} = (x_{1}, \dots, x_{m})$ and $\mathbf{y} = (y_{1}, \dots, y_{m})$ be vectors in $\mathbb{R}^{m}$. We denote as $x_{1}^{\downarrow} \geq \cdots \geq x_{m}^{\downarrow}$ the coordinates of $\mathbf{x}$ arranged in decreasing order. If
\begin{equation} \label{E:majorization}
\sum_{j=1}^{k} x_{j}^{\downarrow} \leq \sum_{j=1}^{k} y_{j}^{\downarrow}
\end{equation}
for all $k = 1, \dots, m$ then we say that $\mathbf{x}$ is weakly majorized by $\mathbf{y}$ and write as $\mathbf{x} \prec_{w} \mathbf{y}$. Additionally if the equality of \eqref{E:majorization} holds for $k = m$, then we say that $\mathbf{x}$ is majorized by $\mathbf{y}$ and write as $\mathbf{x} \prec \mathbf{y}$.

Assume that $\mathbf{x} = (x_{1}, \dots, x_{m})$ and $\mathbf{y} = (y_{1}, \dots, y_{m})$ are vectors with positive entries. We say that $\mathbf{x}$ is weakly log-majorized by $\mathbf{y}$, written as $\mathbf{x} \prec_{w \log} \mathbf{y}$, if
\begin{equation} \label{E:log_majorization}
\prod_{j=1}^{k} x_{j}^{\downarrow} \leq \prod_{j=1}^{k} y_{j}^{\downarrow}
\end{equation}
for all $k = 1, \dots, m$. We say that $\mathbf{x}$ is log-majorized by $\mathbf{y}$, written as $\mathbf{x} \prec_{\log} \mathbf{y}$, additionally if the equality of \eqref{E:log_majorization} holds for $k = m$. One can see that $\mathbf{x} \prec_{(w) \log} \mathbf{y}$ if and only if $\log \mathbf{x} \prec_{(w)} \log \mathbf{y}$, where $\log \mathbf{x} := (\log x_{1}, \dots, \log x_{m})$. It has been known from \cite[Theorem 10.15]{Zh} that $\mathbf{x} \prec_{w \log} \mathbf{y}$ implies $\mathbf{x} \prec_{w} \mathbf{y}$.

\begin{corollary} \label{C:log-majorization}
For $0 < \alpha \leq z < 1$,
\begin{displaymath}
\begin{split}
\lambda( \mathcal{A}(\omega; \mathbb{A}^{1 - \alpha} ) & \prec_{w \log} \lambda( \mathcal{R}_{\alpha, z}(\omega; \mathbb{A}) ), \\
\lambda( P_{z}(\omega; \mathbb{A}^{\frac{1 - \alpha}{z}}) ) & \prec_{w \log} \lambda( \mathcal{R}_{\alpha, z}(\omega; \mathbb{A}) ),
\end{split}
\end{displaymath}
where $\lambda(A)$ denotes the $m$-tuple of eigenvalues of an $m \times m$ matrix $A$. Furthermore, $\lambda( \Lambda(\omega; \mathbb{A}^{\frac{1 - \alpha}{z}}) ) \prec_{w \log} \lambda( \mathcal{R}_{\alpha, z}(\omega; \mathbb{A}) )$.
\end{corollary}

\begin{proof}
Note that the weighted arithmetic mean $\mathcal{A}$ and power mean $P_{z}$ are obviously homogeneous, and the $\alpha-z$ weighted right mean $\mathcal{R}_{\alpha, z}$ is also homogeneous from Lemma \ref{L:properties} (2). So it suffices to show that
\begin{center}
$\mathcal{A}(\omega; \mathbb{A}^{1 - \alpha}) \leq I$ \ and \ $P_{z}(\omega; \mathbb{A}^{\frac{1 - \alpha}{z}}) \leq I$,
\end{center}
whenever $\mathcal{R}_{\alpha, z}(\omega; \mathbb{A}) \leq I$. From Theorem \ref{T:inequalities-2} and Theorem \ref{T:Renyi-power} we obtain the main consequences.

Since $\Lambda \leq P_{z}$ for $0 < z < 1$ from \eqref{E:para-monotonicity}, we obtain the second assertion.
\end{proof}

\begin{remark}
The matrix norm $||| \cdot |||$ on $M_{m}$ is said to be unitarily invariant if $||| U A V ||| = ||| A |||$ for any matrix $A \in M_{m}$ and unitary matrices $U, V$. There is a crucial relation between the weak majorization and unitarily invariant norm of matrices. Precisely for any $A, B \in M_{m}$,
\begin{center}
$\sigma(A) \prec_{w} \sigma(B)$ \ if and only if \ $||| A ||| \leq ||| B |||$
\end{center}
for any unitarily invariant matrix norm $||| \cdot |||$, where $\sigma(A) = (\sigma_{1}(A), \dots, \sigma_{m}(A))$ denotes the $m$-tuple of singular values of $A$. Since $\lambda(A) = \sigma(A)$ for any $A \in \mathbb{P}_{m}$, we have from \cite[Theorem 10.15, Theorem 10.38]{Zh} and Corollary \ref{C:log-majorization}
\begin{center}
$\displaystyle ||| \mathcal{A}(\omega; \mathbb{A}^{1 - \alpha}) ||| \leq ||| \mathcal{R}_{\alpha, z}(\omega; \mathbb{A}) |||$ \ and \ $\displaystyle ||| P_{z}(\omega; \mathbb{A}^{\frac{1 - \alpha}{z}}) ||| \leq ||| \mathcal{R}_{\alpha, z}(\omega; \mathbb{A}) |||$.
\end{center}
\end{remark}

\begin{remark}
Let $\mathbb{A} = (A_{1}, \dots, A_{n}) \in \mathbb{P}_{m}^{n}$, and let $0 < \alpha \leq z < 1$. Then there exist positive scalars $a, b$ such that $a I \leq A_{j} \leq b I$ for all $j$. So $b^{-1} A_{j} \leq I$ implies $\mathcal{R}_{\alpha, z}(\omega; b^{-1} \mathbb{A}) \leq I$ by Lemma \ref{L:boundedness}, and similarly, $a^{-1} A_{j} \geq I$ implies $\mathcal{R}_{\alpha, z}(\omega; a^{-1} \mathbb{A}) \geq I$. Thus, we obtain the modified consequences of Theorem \ref{T:inequalities-2} and Theorem \ref{T:Renyi-power} as follows:
\begin{displaymath}
\begin{split}
b^{\alpha(1 - \frac{1}{z})} \mathcal{A}(\omega; \mathbb{A}^{1 - \alpha}) \leq & \, \mathcal{R}_{\alpha, z}(\omega; \mathbb{A})^{1 - \frac{\alpha}{z}} \leq a^{\alpha (1 - \frac{1}{z})} \mathcal{A}(\omega; \mathbb{A}^{1 - \alpha}), \\
b^{1 - \frac{1}{z}} P_{z}(\omega; \mathbb{A}^{\frac{1 - \alpha}{z}}) \leq & \, \mathcal{R}_{\alpha, z}(\omega; \mathbb{A})^{1 - \frac{\alpha}{z}} \leq a^{1 - \frac{1}{z}} P_{z}(\omega; \mathbb{A}^{\frac{1 - \alpha}{z}}).
\end{split}
\end{displaymath}
By the monotonicity of matrix power means for parameters in \eqref{E:para-monotonicity} we obtain
\begin{displaymath}
b^{1 - \frac{1}{z}} \mathcal{H}(\omega; \mathbb{A}^{\frac{1 - \alpha}{z}}) \leq \mathcal{R}_{\alpha, z}(\omega; \mathbb{A})^{1 - \frac{\alpha}{z}} \leq a^{1 - \frac{1}{z}} \mathcal{A}(\omega; \mathbb{A}^{\frac{1 - \alpha}{z}}).
\end{displaymath}
\end{remark}

\section{Trace inequality with Wasserstein mean}

From the $L_{2}$-Wasserstein distance of Gaussian distributions with mean zero and covariance matrices $A, B \in \mathbb{P}_{m}$, a new metric on $\mathbb{P}_{m}$ has been introduced \cite{BJL}:
\begin{displaymath}
d_{W}(A, B) = \left[ \tr \left( \frac{A + B}{2} \right) - \tr (A^{1/2} B A^{1/2})^{1/2} \right]^{1/2}.
\end{displaymath}
This is called the Bures-Wasserstein metric, and note that $d_{W}(A, B) = \Phi_{\frac{1}{2}, \frac{1}{2}} (A, B)$. It coincides with the Bures distance of density matrices in quantum information theory and is the matrix version of Hellinger distance.

Let $\mathbb{A} = (A_{1}, \dots, A_{n}) \in \mathbb{P}_{m}^{n}$, and let $\omega = (w_{1}, \dots, w_{n}) \in \Delta_{n}$. We consider the following minimization problem
\begin{equation} \label{E:minimization-Wass}
\underset{X \in \mathbb{P}_{m}}{\arg \min} \sum_{j=1}^{n} w_{j} d_{W}^{2}(X, A_{j}).
\end{equation}
By using tools in non-smooth analysis, convex duality, and optimal transport theory, it has been proved in \cite[Theorem 6.1]{AC} that the above minimization problem \eqref{E:minimization-Wass} has a unique solution in $\mathbb{P}_{m}$. On the other hand, it has been shown in \cite{BJL} that the objective function $\displaystyle g(X) = \sum_{j=1}^{n} w_{j} d_{W}^{2}(X, A_{j})$ is strictly convex on $\mathbb{P}_{m}$, by applying the strict concavity of the map $h: \mathbb{P}_{m} \to \mathbb{R}, \ h(X) = \tr (X^{1/2})$. Therefore, we define the \emph{Wasserstein mean} $\Omega(\omega; \mathbb{A})$ as such a unique minimizer of \eqref{E:minimization-Wass}. Note from the definition of the $\alpha-z$ weighted right mean that
\begin{displaymath}
\Omega(\omega; \mathbb{A}) = \underset{X \in \mathbb{P}_{m}}{\arg \min} \sum_{j=1}^{n} w_{j} d_{W}^{2}(X, A_{j}) = \mathcal{R}_{\frac{1}{2}, \frac{1}{2}}(\omega; \mathbb{A}).
\end{displaymath}

An iteration approach to the Wasserstein mean $\Omega(\omega; \mathbb{A})$ has been recently shown in \cite{ABCM} by using the map $K: \mathbb{P}_{m} \to \mathbb{P}_{m}$ defined as
\begin{equation} \label{E:Wass}
K(A) = A^{-1/2} \left[ \sum_{j=1}^{n} w_{j} \left(A^{1/2} A_{j} A^{1/2}\right)^{1/2} \right]^{2} A^{-1/2}
\end{equation}
for each $A \in \mathbb{P}_{m}$.

\begin{theorem} \cite{ABCM} \label{T:iteration}
Let $\omega \in \Delta_{n}$ and $\mathbb{A} \in \mathbb{P}_{m}^{n}$. For every $S_{0} \in \mathbb{P}_{m}$ the sequence $S_{r+1} = K(S_{r})$ constructed iteratively from the map $K$ in \eqref{E:Wass} converges to $\Omega(\omega; \mathbb{A})$, and for all natural numbers $r$
\begin{displaymath}
\tr S_{r} \leq \tr S_{r+1} \leq \tr \Omega(\omega; \mathbb{A}).
\end{displaymath}
\end{theorem}

\begin{theorem} \label{T:Wass-Renyi}
For $1 \leq p < 2$
\begin{displaymath}
\tr \mathcal{R}_{1 - \frac{p}{2}, \frac{1}{2}}(\omega; \mathbb{A})^{p} \leq \tr \Omega(\omega; \mathbb{A}^{p}).
\end{displaymath}
\end{theorem}

\begin{proof}
Let $\alpha := 1 - \frac{p}{2}$ and $X = \mathcal{R}_{\alpha, \frac{1}{2}}(\omega; \mathbb{A})$. Then $p = 2 (1 - \alpha)$ and $0 < \alpha \leq \frac{1}{2}$. By \eqref{E:Renyi}
\begin{displaymath}
X = \sum_{j=1}^{n} w_{j} \left(X^{\alpha} A_{j}^{2(1 - \alpha)} X^{\alpha}\right)^{\frac{1}{2}} = \sum_{j=1}^{n} w_{j} \left(\left(X^{2 \alpha}\right)^{\frac{1}{2}} A_{j}^{p} \left(X^{2 \alpha}\right)^{\frac{1}{2}}\right)^{\frac{1}{2}}.
\end{displaymath}
Taking the square map and congruence transformation by $X^{- \alpha}$ on both sides yields
\begin{displaymath}
X^{2(1 - \alpha)} = \left(X^{2 \alpha}\right)^{- \frac{1}{2}} \left[ \sum_{j=1}^{n} w_{j} \left(\left(X^{2 \alpha}\right)^{\frac{1}{2}} A_{j}^{p} \left(X^{2 \alpha}\right)^{\frac{1}{2}}\right)^{\frac{1}{2}} \right]^{2} (X^{2 \alpha})^{- \frac{1}{2}} = K\left(X^{2 \alpha}\right).
\end{displaymath}
Note from Theorem \ref{T:iteration} that the sequence $S_{r}$ constructed by the map $K$ in \eqref{E:Wass} with the initial value $S_{0} = X^{2 \alpha}$ converges to $\Omega(\omega; \mathbb{A}^{p})$, and furthermore,
\begin{displaymath}
\tr X^{2(1 - \alpha)} = \tr K\left(X^{2 \alpha}\right) \leq \tr \Omega(\omega; \mathbb{A}^{p}),
\end{displaymath}
which gives the desired inequality.
\end{proof}

\begin{remark}
The consequence of Theorem \ref{T:Wass-Renyi} can be rewritten as
\begin{displaymath}
\tr \mathcal{R}_{\alpha, \frac{1}{2}}(\omega; \mathbb{A})^{2 (1 - \alpha)} \leq \tr \Omega(\omega; \mathbb{A}^{2 (1 - \alpha)}).
\end{displaymath}
Furthermore, we have from Theorem \ref{T:Wass-Renyi}
\begin{displaymath}
\tr \mathcal{R}_{1 - \frac{p}{2^{k}}, \frac{1}{2}}(\omega; \mathbb{A}^{2^{k-1}})^{\frac{p}{2^{k-1}}} \leq \tr \Omega(\omega; \mathbb{A}^{p})
\end{displaymath}
for any $p \geq 1$ such that $2^{k-1} \leq p < 2^{k}$ for some $k \in \mathbb{N}$.
\end{remark}

\section{Hadamard product}

The \emph{Hadamard} (or \emph{Schur}) \emph{product} $A \circ B$ of $A = [a_{ij}]$ and $B = [b_{ij}]$ in $M_{m, k}$ is the $m \times k$ matrix, which is defined by the entrywise product:
\begin{displaymath}
\displaystyle A \circ B := [a_{ij} b_{ij}].
\end{displaymath}
Note that Hadamard product is bilinear, commutative, and associative. Furthermore, the Hadamard product gives us a binary operation on $M_{m, k}$. Moreover, the Hadamard product preserves positivity; the Hadamard product of two positive definite (positive semidefinite, respectively) matrices is positive definite (positive semidefinite, respectively) matrices. This is known as the Schur product theorem \cite{HJ,Zh}.

The \emph{tensor} (or \emph{Kronecker}) \emph{product} $A \otimes B$ of $A = [a_{ij}] \in M_{m,k}$ and $B = [b_{ij}] \in M_{s,t}$ is the $ms \times kt$ matrix given by
\begin{displaymath}
A \otimes B :=
\left[
  \begin{array}{ccc}
    a_{11} B & \cdots & a_{1k} B \\
    \vdots & \ddots & \vdots \\
    a_{m1} B & \cdots & a_{mk} B \\
  \end{array}
\right] \in M_{ms, kt}.
\end{displaymath}
Note that the tensor product is bilinear and associative, but not commutative, see \cite{Bh,Zh}. There is a canonical relationship between the tensor product and Hadamard product via a positive unital linear map.
\begin{lemma}\cite[Lemma 4]{An} \label{L:An}
There exists a strictly positive and unital linear map $\Psi$ such that for any $A, B \in M_{m}$
\begin{displaymath}
\Psi(A \otimes B)= A \circ B.
\end{displaymath}
\end{lemma}

For convenience, we denote as for $\mathbb{A} = (A_{1}, \dots, A_{n}), \mathbb{B} = (B_{1}, \dots, B_{n}) \in \mathbb{P}_{m}^{n}$
\begin{displaymath}
\begin{split}
\mathbb{A} \otimes \mathbb{B} & := (\underbrace{A_{1} \otimes B_{1}, \dots, A_{1} \otimes B_{n}}, \dots, \underbrace{A_{n} \otimes B_{1}, \dots, A_{n} \otimes B_{n}}), \\
\mathbb{A} \circ \mathbb{B} & := (\underbrace{A_{1} \circ B_{1}, \dots, A_{1} \circ B_{n}}, \dots, \underbrace{A_{n} \circ B_{1}, \dots, A_{n} \circ B_{n}}).
\end{split}
\end{displaymath}
From Lemma \ref{L:equation} or \eqref{E:Renyi} one can easily obtain the following identity for the tensor product of $\alpha-z$ weighted right means.

\begin{theorem} \cite[Theorem 3.2]{HJK}\label{T:Tensor}
Let $\mathbb{A}, \mathbb{B} \in \mathbb{P}_{m}^{n}$, and let $\omega = (w_{1}, \dots, w_{n}), \mu = (\mu_{1}, \dots, \mu_{n}) \in \Delta_{n}$. Then
\begin{displaymath}
\mathcal{R}_{\alpha, z}(\omega; \mathbb{A}) \otimes \mathcal{R}_{\alpha, z}(\mu; \mathbb{B}) = \mathcal{R}_{\alpha, z}(\omega \otimes \mu; \mathbb{A} \otimes \mathbb{B})
\end{displaymath}
where $\omega \otimes \mu := (\underbrace{w_{1} \mu_{1}, \dots, w_{1} \mu_{n}}, \dots, \underbrace{w_{n} \mu_{1}, \dots, w_{n} \mu_{n}}) \in \Delta_{n^{2}}$.
\end{theorem}

Now, we show the bounds for the Hadamard product of $\alpha-z$ weighted right means.

\begin{theorem} \label{T:Hada1}
Let $\mathbb{A}, \mathbb{B} \in \mathbb{P}_{m}^{n}$ and let $\omega, \mu \in \Delta_{n}$. Then for $\frac{1}{2} \leq z < 1$
\begin{displaymath}
\mathcal{R}_{\alpha,z}(\omega; \mathbb{A})^{\frac{1-\alpha}{z}} \circ \mathcal{R}_{\alpha,z}(\mu;\mathbb{B})^{\frac{1-\alpha}{z}} \leq \mathcal{A}(\omega \otimes \mu;\mathbb{A}^{\frac{1-\alpha}{z}}\circ\mathbb{B}^{\frac{1-\alpha}{z}}).
\end{displaymath}
\end{theorem}

\begin{proof}
Using Lemma \ref{L:An}, we get
\begin{displaymath}
\begin{split}
\mathcal{R}_{\alpha,z}(\omega; \mathbb{A})^{\frac{1-\alpha}{z}} \circ \mathcal{R}_{\alpha,z}(\mu;\mathbb{B})^{\frac{1-\alpha}{z}} & = \Psi\left(\mathcal{R}_{\alpha,z}(\omega; \mathbb{A})^{\frac{1-\alpha}{z}} \otimes \mathcal{R}_{\alpha,z}(\mu;\mathbb{B})^{\frac{1-\alpha}{z}} \right) \\
& = \Psi\left(\left(\mathcal{R}_{\alpha,z}(\omega; \mathbb{A}) \otimes \mathcal{R}_{\alpha,z}(\mu; \mathbb{B})\right)^{\frac{1-\alpha}{z}}\right) \\
& = \Psi\left(\mathcal{R}_{\alpha,z}(\omega \otimes \mu; \mathbb{A} \otimes \mathbb{B})^{\frac{1-\alpha}{z}}\right).
\end{split}
\end{displaymath}
The second equality follows from the property of tensor product, and the third follows from Theorem \ref{T:Tensor}. Moreover, applying Theorem \ref{T:A-R ineq} to the last equality and using Lemma \ref{L:An}, we obtain
\begin{displaymath}
\begin{split}
\Psi\left(\mathcal{R}_{\alpha,z}(\omega \otimes \mu; \mathbb{A} \otimes \mathbb{B})^{\frac{1-\alpha}{z}}\right)
& \leq \Psi \left( \mathcal{A}(\omega \otimes \mu; (\mathbb{A} \otimes \mathbb{B})^{\frac{1-\alpha}{z}})\right)\\
& = \mathcal{A}(\omega \otimes \mu; \Psi(\mathbb{A}^{\frac{1-\alpha}{z}} \otimes \mathbb{B}^{\frac{1-\alpha}{z}}))
= \mathcal{A}(\omega \otimes \mu; \mathbb{A}^{\frac{1-\alpha}{z}}\circ\mathbb{B}^{\frac{1-\alpha}{z}}).
\end{split}
\end{displaymath}
It completes the proof.
\end{proof}

\begin{theorem} \label{T:Hada2}
Let $\mathbb{A}, \mathbb{B} \in \mathbb{P}_{m}^{n}$ and let $\omega, \mu \in \Delta_{n}$. If $\mathcal{R}_{\alpha,z}(\omega; \mathbb{A}) \geq I$ and $\mathcal{R}_{\alpha,z}(\mu; \mathbb{B}) \geq I$ then
\begin{displaymath}
\mathcal{R}_{\alpha,z}(\omega; \mathbb{A})^{1-\frac{\alpha}{z}} \circ \mathcal{R}_{\alpha,z}(\mu; \mathbb{B})^{1-\frac{\alpha}{z}}
\leq P_{z}(\omega \otimes \mu; \mathbb{A}^{\frac{1-\alpha}{z}} \circ \mathbb{B}^{\frac{1-\alpha}{z}}),
\end{displaymath}
where $P_{z}$ denotes the matrix power mean for $z \in (0,1)$.
\end{theorem}

\begin{proof}
By Lemma \ref{L:An} and Theorem \ref{T:Tensor}, we have
\begin{displaymath}
\begin{split}
\mathcal{R}_{\alpha,z}(\omega; \mathbb{A})^{1 - \frac{\alpha}{z}} \circ \mathcal{R}_{\alpha,z}(\mu; \mathbb{B})^{1 - \frac{\alpha}{z}} &= \Psi \left( \mathcal{R}_{\alpha,z}(\omega; \mathbb{A})^{1 - \frac{\alpha}{z}} \otimes \mathcal{R}_{\alpha,z}(\mu; \mathbb{B})^{1 - \frac{\alpha}{z}} \right) \\
&= \Psi \left( \mathcal{R}_{\alpha,z}(\omega \otimes \mu; \mathbb{A} \otimes \mathbb{B})^{1 - \frac{\alpha}{z}} \right).
\end{split}
\end{displaymath}
Note that $\mathcal{R}_{\alpha,z}(\omega \otimes \mu ; \mathbb{A} \otimes \mathbb{B}) \geq I,$ since $\mathcal{R}_{\alpha,z}(\omega; \mathbb{A}) \geq I$ and $\mathcal{R}_{\alpha,z}(\mu; \mathbb{B}) \geq I$. So applying Theorem \ref{T:Renyi-power}, we obtain
\begin{displaymath}
\begin{split}
\Psi \left( \mathcal{R}_{\alpha,z}(\omega \otimes \mu; \mathbb{A} \otimes \mathbb{B})^{1 - \frac{\alpha}{z}} \right)
& \leq \Psi \left( P_{z}(\omega \otimes \mu; (\mathbb{A} \otimes \mathbb{B})^{\frac{1-\alpha}{z}}) \right) \\
& \leq P_{z}(\omega \otimes \mu; \Psi(\mathbb{A}^{\frac{1-\alpha}{z}} \otimes \mathbb{B}^{\frac{1-\alpha}{z}}))
= P_{z}(\omega \otimes \mu; \mathbb{A}^{\frac{1-\alpha}{z}} \circ \mathbb{B}^{\frac{1-\alpha}{z}}).
\end{split}
\end{displaymath}
The second inequality follows from \cite[Proposition 3.5]{LP}.
\end{proof}

\begin{remark}
By the monotonicity of matrix power means for parameters in \eqref{E:para-monotonicity} we have from Theorem \ref{T:Hada2} that
\begin{displaymath}
\mathcal{R}_{\alpha,z}(\omega; \mathbb{A})^{1 - \frac{\alpha}{z}} \circ \mathcal{R}_{\alpha,z}(\mu; \mathbb{B})^{1 - \frac{\alpha}{z}} \leq P_{z}(\omega \otimes \mu; \mathbb{A}^{\frac{1-\alpha}{z}} \circ \mathbb{B}^{\frac{1-\alpha}{z}}) \leq \mathcal{A}(\omega \otimes \mu; \mathbb{A}^{\frac{1-\alpha}{z}} \circ \mathbb{B}^{\frac{1-\alpha}{z}}).
\end{displaymath}
\end{remark}

\begin{corollary}
Let $\mathbb{A} = (A_{1}, \dots, A_{n}), \mathbb{B} = (B_{1}, \dots, B_{n}) \in \mathbb{P}_{m}^{n}$ satisfying $A_{j} \geq a I$ and $B_{j} \geq b I$ for all $j$ and some $a, b > 0$. Then
\begin{displaymath}
\mathcal{R}_{\alpha,z}(\omega; \mathbb{A})^{1 - \frac{\alpha}{z}} \circ \mathcal{R}_{\alpha,z}(\mu; \mathbb{B})^{1 - \frac{\alpha}{z}} \leq (a b)^{1 - \frac{1}{z}} P_{z}(\omega \otimes \mu; \mathbb{A}^{\frac{1-\alpha}{z}} \circ \mathbb{B}^{\frac{1-\alpha}{z}}).
\end{displaymath}
\end{corollary}

\begin{proof}
Since $A_{j} \geq a I$ for all $j$ and some $a > 0$, we have $\mathcal{R}_{\alpha,z}(\omega; \mathbb{A}) \geq a I$ from Lemma \ref{L:boundedness}. By homogeneity of the $\alpha-z$ weighted right mean in Lemma \ref{L:properties} (2), $\mathcal{R}_{\alpha,z}(\omega; \frac{1}{a} \mathbb{A}) \geq I$. Similarly, $\mathcal{R}_{\alpha,z}(\mu; \frac{1}{b} \mathbb{B}) \geq I$. By Theorem \ref{T:Hada2} we have
\begin{displaymath}
\mathcal{R}_{\alpha,z} \left( \omega; \frac{1}{a} \mathbb{A} \right)^{1 - \frac{\alpha}{z}} \circ \mathcal{R}_{\alpha,z} \left( \mu; \frac{1}{b} \mathbb{B} \right)^{1 - \frac{\alpha}{z}} \leq P_{z} \left( \omega \otimes \mu; \left( \frac{1}{a} \mathbb{A} \right)^{\frac{1-\alpha}{z}} \circ \left( \frac{1}{b} \mathbb{B} \right)^{\frac{1-\alpha}{z}} \right).
\end{displaymath}
Equivalently, by using the homogeneity of the $\alpha-z$ weighted right mean and the matrix power mean
\begin{displaymath}
\left( \frac{1}{a b} \right)^{1 - \frac{\alpha}{z}} \mathcal{R}_{\alpha,z}(\omega; \mathbb{A})^{1 - \frac{\alpha}{z}} \circ \mathcal{R}_{\alpha,z}(\mu; \mathbb{B})^{1 - \frac{\alpha}{z}}
\leq \left( \frac{1}{a b} \right)^{\frac{1-\alpha}{z}} P_{z}(\omega \otimes \mu; \mathbb{A}^{\frac{1-\alpha}{z}} \circ \mathbb{B}^{\frac{1-\alpha}{z}}).
\end{displaymath}
By a simple calculation we obtain the desired inequality.
\end{proof}



\section{Summary and final remark}

We have shown in this paper many interesting properties of $\alpha-z$ weighted right mean such as operator inequalities with the matrix power means, trace inequality with the Wasserstein mean, and inequalities in terms of Hadamard product. We can arrange the consequences with the following figure:
\begin{figure}[h!]
\centering
\includegraphics[scale=0.3]{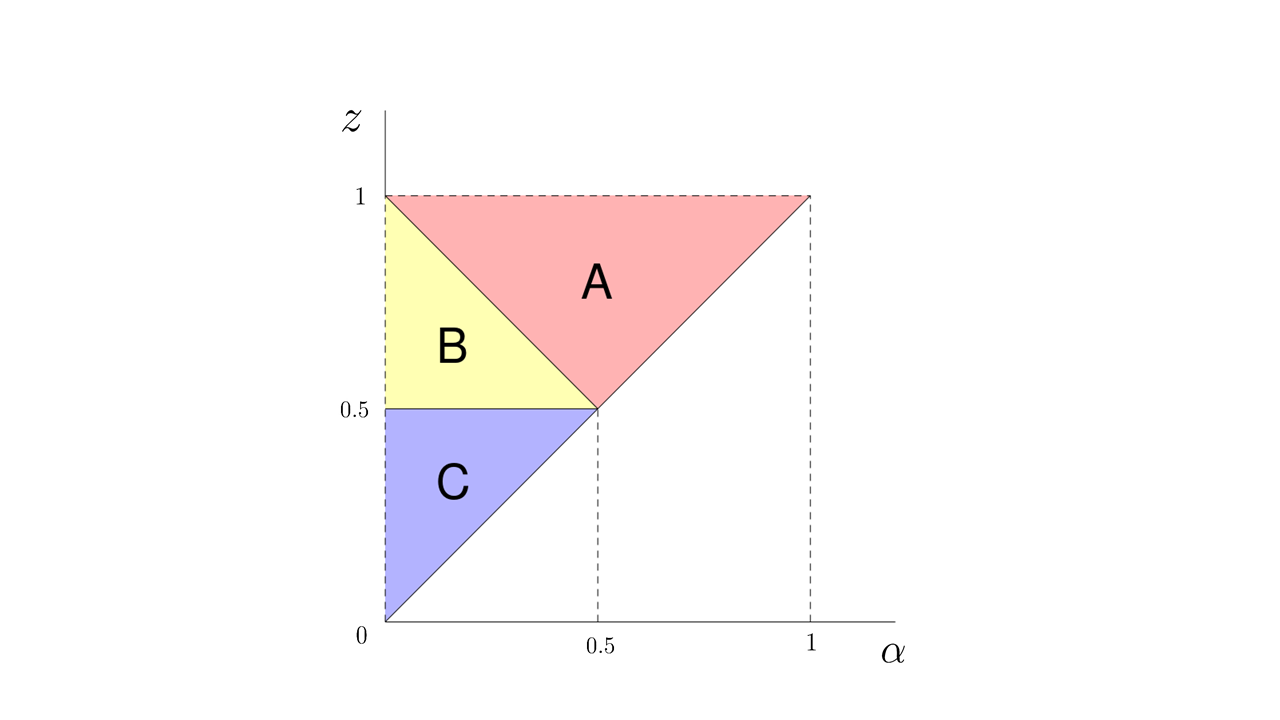}
\caption{Regions partitioned by properties of the $\alpha-z$ weighted right mean} \label{fig1}
\end{figure}

In Figure \ref{fig1} we denote as
\begin{displaymath}
\begin{split}
\textrm{A} & = \{ (\alpha, z) \in (0,1)^{2}: z \geq \max \{ 1 - \alpha, \alpha \} \}, \\
\textrm{B} & = \{ (\alpha, z) \in (0,1)^{2}: 1/2 \leq z \leq 1 - \alpha \}, \\
\textrm{C} & = \{ (\alpha, z) \in (0,1)^{2}: \alpha \leq z \leq 1/2 \}.
\end{split}
\end{displaymath}
The $\alpha-z$ weighted right mean $\mathcal{R}_{\alpha, z}$ is defined on the region $\textrm{A} \cup \textrm{B} \cup \textrm{C}$, and Theorem \ref{T:inequalities-2}, Theorem \ref{T:Renyi-power}, and Theorem \ref{T:Hada2} are satisfied on the same region. Finally, Theorem \ref{T:A-R ineq} and Theorem \ref{T:Hada1} hold on the region $\textrm{A} \cup \textrm{B}$, and Theorem \ref{T:Wass-Renyi} is satisfied on the region $\textrm{B} \cap \textrm{C}$.

The quantum divergence $\Phi_{\alpha, z}$ is not symmetric, that is, $\Phi_{\alpha, z}(A, B) \neq \Phi_{\alpha, z}(B, A)$ for $A, B \in \mathbb{P}_{m}$ in general. So one may be interested in the left mean
\begin{displaymath}
\underset{X \in \mathbb{P}_{m}}{\arg \min} \sum_{j=1}^{n} w_{j} \Phi_{\alpha, z}(X, A_{j}).
\end{displaymath}
It has been known neither the divergence function
\begin{displaymath}
\Phi_{\alpha, z}(X, A) = \tr ((1 - \alpha) X + \alpha A) - \tr  \left( X^{\frac{1 - \alpha}{2z}} A^{\frac{\alpha}{z}} X^{\frac{1 - \alpha}{2z}} \right)^{z}
\end{displaymath}
for given $A \in \mathbb{P}_{m}$ is strictly convex nor the above minimization can be solved. So the existence and uniqueness for the solution of the above minimization would be an interesting topic, and we can find properties of the left mean analogous to the right mean upon success.

\vspace{1cm}

\textbf{Acknowledgement}

This work was supported by the National Research Foundation of Korea (NRF) grant funded by the Korea government (MSIT) (No. NRF-2018R1C1B6001394).

\end{document}